\documentclass[12pt]{amsart}
\usepackage{amssymb}
\usepackage{fullpage}
\usepackage{graphicx}
\usepackage{enumerate}
\input pictex.sty

\newtheorem{thm}{Theorem}[section]

\newtheorem{lem}[thm]{Lemma}

\theoremstyle{definition}

\theoremstyle{remark}

\numberwithin{equation}{section}

\newcommand{\p}{\partial}
\newcommand{\no}{\nonumber}

\newcounter{space}

\begin{document}

\title{Self-Similar Solutions of the Non-Strictly Hyperbolic Whitham Equations for the KdV Hierarchy}%
\author{V. U. Pierce}
\address{Department of Mathematics, Ohio State University, 231 W. 18th Avenue, Columbus, OH 43210}%
\email{vpierce@math.ohio-state.edu}
\author{Fei-Ran Tian}%
\address{Department of Mathematics, Ohio State University, 231 W. 18th Avenue, Columbus, OH 43210}
\email{tian@math.ohio-state.edu}

\thanks{}%

\begin{abstract}
We study the Whitham equations for 
all the higher order KdV equations. The Whitham equations are neither strictly hyperbolic nor
genuinely nonlinear. We are interested in the solution of the Whitham equations
when the initial values are given by a step function. 
\end{abstract}
\maketitle

\markboth{ V.U. PIERCE AND F.-R. TIAN}{
SELF-SIMILAR SOLUTION OF THE WHITHAM EQUATIONS}
\pagestyle{myheadings}

\section{Introduction}

It is known that the solution of the KdV equation
\begin{equation} 
\label{eq1}
u_t + 6 u u_x + \epsilon^2 u_{xxx}  =  0 
\end{equation} 
has a weak limit as $\epsilon \to 0$ while the initial values 
$$u(x, 0; \epsilon) = u_0(x)$$
are fixed.

This weak limit is described by hyperbolic equations. It satisfies the Burgers equation
\begin{equation}
\label{Burgers}
u_t + (3 u^2)_x = 0
\end{equation}
until its solution develops shocks. Immediately after shock, the weak limit is governed
by the Whitham equations \cite{lax, lax2, ven, whi} 
\begin{equation}
\label{KdVW}
u_{it} + \lambda_i(u_1, u_2, u_3) u_{ix} = 0 \ , \quad i=1, 2, 3,
\end{equation}
where the $\lambda_i$'s are given by formulae (\ref{lambda}).
Equations (\ref{KdVW}) form a $3 \times 3$ system of hyperbolic equations. 
After the breaking of the solution of (\ref{KdVW}), the weak limit is described by a $5 \times 5$
systems of
hyperbolic equations similar to (\ref{KdVW}). Similarly, after the solution of the $5 \times 5$
system breaks down, the weak limit is characterized by a $7 \times 7$ system of
hyperbolic equations. In other words, for general initial data $u_0(x)$, one constructs
the weak limit by patching together solutions of (\ref{Burgers}), (\ref{KdVW}), $5 \times 5$,
$7 \times 7$, etc systems in the $x$-$t$ plane.

The KdV equation (\ref{eq1}) is just the first of an infinite sequence of equations.
All these so-called higher order KdV equations can be
cast in
the Hamiltonian form
\begin{equation}
\label{hierarchy}
\frac{\partial
u}{\partial
t} + \frac{\partial}{\partial x } \frac{\delta H_m}{\delta u} = 0 
\ , \quad  m=1, 2,
\cdots \ ,
\end{equation}
where $H_m$'s form a sequence of conserved
functionals for
the KdV
equation. The small dispersive parameter $\epsilon$ is hidden in $H_m$.   
In particular, when $m=1$, (\ref{hierarchy}) is the KdV equation.

The solution of the higher order KdV equation (\ref{hierarchy}) also has a weak limit as $\epsilon \to 0$.
As in the KdV case, this weak limit satisfies the Burgers type
equation 
\begin{equation}
\label{mBurgers}
u_t + ({C_m \over m + 1} u^{m+1})_x = 0 \ , 
\end{equation}
where $C_m$ is given in (\ref{Cm}), until the solution of (\ref{mBurgers}) forms a shock. 
After the breaking of the solution of (\ref{mBurgers}),
the limit is governed by equations similar
to (\ref{KdVW}), namely,
\begin{equation}
\label{mKdVW}
u_{it} + \mu_i^{(m)}(u_1, u_2, u_3) u_{ix} = 0 \ , \quad i=1, 2, 3,
\end{equation}
where $\mu_i^{(m)}$'s are given in (\ref{eq18a}). They will also be called the Whitham equations.
As in the KdV case, after the solution of (\ref{mKdVW}) breaks down, the weak limit is described
by a $5 \times 5$ system of hyperbolic equations.

In this paper, we are interested in the solution of the Whitham equations
for the higher order KdV (\ref{hierarchy}) with a step-like initial function
\begin{equation} \label{step}
u_0(x) = \left\{ \begin{matrix} 1 & x < 0 \\
0 & x > 0 \ .\end{matrix} \right.  
\end{equation}
For such an initial function, the solution of the Burgers type 
equation (\ref{mBurgers}) has
already developed a shock at the initial time, $t=0$.
Hence, immediately
after $t=0$, the Whitham equations (\ref{mKdVW}) kick in.  Solutions of
(\ref{mKdVW}) occupy some domains of the space-time while solutions of
(\ref{mBurgers}) occupy other domains.  These solutions are matched on
the boundaries of the domains.

Equations (\ref{Burgers}) and equations (\ref{mBurgers}) are prototypes in the theory of hyperbolic 
conservation laws \cite{lef}. Their solutions will generally develop shocks in finite times. 
The solutions can be extended beyond the singularities as the entropy solutions.    

Solutions of equations (\ref{Burgers}) or equations (\ref{mBurgers}), in the theory of the zero dispersion limit,
are not extended as weak or entropy solutions after the formation of singularities. Instead, they are extended to 
match the Whitham solutions of (\ref{KdVW}) or (\ref{mKdVW}). For initial data (\ref{step}),
the resulting solutions of the Whitham equations (\ref{mKdVW}) will be seen to be more complex than those
of (\ref{KdVW}) in the KdV case. 

The KdV case with the step-like initial data (\ref{step}) was first studied by Gurevich and Pitaevskii
\cite{gur}.  
They found that it was enough to use the Burgers solution of (\ref{Burgers}) and the Whitham solution of
(\ref{KdVW}) to cover the whole $x$-$t$ plane, without going to the $5 \times 5$ or $7 \times 7$ system. 
Namely, the space-time is divided into three parts 
$$(1) \   \frac{x}{t} < -6 \ , \quad
(2) \   -6 < \frac{x}{t} < 4  \ , \quad
(3) \  \frac{x}{t} > 4 \ .$$ 
The solution of (\ref{Burgers}) occupies the first and third parts, 
\begin{equation} 
\label{eq9a}
u(x, t) \equiv 1  \quad \mbox{when $\frac{x}{t} < -6$} \ , \quad 
u(x, t) \equiv 0 \quad  \mbox{when $\frac{x}{t} > 4$} \ .
\end{equation}
The Whitham solution of (\ref{KdVW}) lives in the second part, 
\begin{equation}
u_1(x, t) \equiv 1 \ , \quad
\frac{x}{t} = \lambda_2(1, u_2, 0) \ , \quad
u_3(x, t) \equiv 0 \ , 
\label{eq9b}
\end{equation}
when $-6 < x/t < 4$.  

Whether the second equation of (\ref{eq9b}) can be inverted to give $u_2$ as a function of the
self-similarity variable $x/t$ hinges on whether 
\begin{equation*}  
\frac{\p \lambda_2}{\p u_2} (1, u_2, 0) \neq 0. 
\end{equation*}
Indeed, Levermore \cite{lev} has proved the genuine nonlinearity of
the Whitham equations (\ref{KdVW}), i.e., 
\begin{equation}\label{eq11}
\frac{\p \lambda_i}{\p u_i} (u_1, u_2, u_3) > 0, \,
\quad i=1, 2, 3, 
\end{equation}
for $u_1 > u_2 > u_3$.  

For the higher order KdV (\ref{hierarchy}),  equations (\ref{mKdVW}), in
general, are not genuinely nonlinear, i.e., a property like
(\ref{eq11}) is not available.  Hence, solutions like (\ref{eq9a}) and (\ref{eq9b})
need to be modified.  

Our construction of solutions of the Whitham
equation (\ref{mKdVW}) makes use of the non-strict
hyperbolicity of the equations.  For KdV, it is known that the Whitham
equations (\ref{KdVW}) are strictly hyperbolic, namely:
\begin{equation*}
\lambda_1(u_1, u_2, u_3) > \lambda_2(u_1, u_2, u_3) > 
\lambda_3(u_1, u_2, u_3) 
\end{equation*}
for $u_1 > u_2 > u_3$ \cite{lev}.  For the higher order KdV (\ref{hierarchy}), different eigenspeeds of (\ref{mKdVW}),
$\mu_i^{(m)}(u_1, u_2, u_3)$'s, may coalesce in the region $u_1 > u_2 > u_3$ \cite{PT}.  

For the higher order KdV with step-like initial function (\ref{step}), the space time is 
divided into four regions (see Figure 1.) 
$$(1) \ \frac{x}{t} < \alpha \ , \quad
(2) \ \alpha < \frac{x}{t} < \beta \ , \quad
(3) \ \beta < \frac{x}{t} < 4^m \ , \quad
(4) \ \frac{x}{t} > 4^m \ ,$$
where $\alpha$ and $\beta$ are some constants.
In the first and fourth regions, the solution of (\ref{mBurgers}) governs
the evolution:
$$u(x, t) \equiv 1 \quad \mbox{where $x/t < \alpha$} \  \mbox{and} \ 
u(x, t) \equiv 0 \quad \mbox{where $x/t > 4^m$} \ .$$
The Whitham solution of (\ref{mKdVW}) lives in the second and third
regions; namely:
\begin{equation} 
\label{ns}
u_1(x, t) \equiv 1 \ , \quad 
\frac{x}{t} = \mu_2^{(m)}(1, u_2, u_3) \ , \quad
\frac{x}{t} = \mu_3^{(m)}(1, u_2, u_3) \ , 
\end{equation}
when $\alpha < x/t < \beta$, and
\begin{equation}
\label{ns2}
u_1(x, t) \equiv 1 \ , \quad 
\frac{x}{t} = \mu_2^{(m)}(1, u_2, 0) \ , \quad  u_3(x, t) \equiv 0 \ , 
\end{equation}
when $\beta < x/t < 4^m$.

Equations (\ref{ns}) yield
$$ \mu_2^{(m)}(1, u_2, u_3) = \mu_3^{(m)}(1, u_2, u_3)$$
on a curve in the region $0< u_3 < u_2 < 1$.
This implies the non-strict hyperbolicity of the Whitham equations (\ref{mKdVW}) for
the KdV hierarchy.

The $m=2$ case has been studied in \cite{PT}. There, inequalities
\begin{equation}
{\p \mu_3^{(m)} \over \p u_3} < {3 \over 2} \ {\mu_2^{(m)} - \mu_3^{(m)} \over u_2 - u_3} < 
{\p \mu_2^{(m)} \over \p u_2} \quad \mbox{for $u_1 > u_2 > u_3 > 0$} \label{INQ}
\end{equation}
have played a crucial role in verifying that equations (\ref{ns}) or (\ref{ns2}) can 
indeed be solved to give the solution of the Whitham equations (\ref{mKdVW}) when 
$m=2$.

For $m>2$, inequalities (\ref{INQ}) are not valid any more. We therefore must
use a different approach to solve the problem. The calculations are considerably more difficult
than in the $m=2$ case. This is mainly because $q$ of (\ref{q}) is a polynomial of degree $m$ when $m >2$
while it is only a quadratic polynomial when $m=2$. 

The organization of the paper is as follows. In Section 2, we will study the
eigenspeeds, $\mu_i^{(m)}$'s, of the Whitham equations (\ref{mKdVW}).
In Section 3, we will construct the self-similar solution of the Whitham equations
for the initial function (\ref{step}).
In Section 4, we will use the self-similar solution of Section 3 to
construct the minimizer of a variational problem for the zero dispersion limit 
of the KdV hierarchy. 

In a subsequent publication, we will study the Whitham solutions for all the other
step-like initial data.


\section{The Whitham Equations}

In this section we define the eigenspeeds of the Whitham equations for
both the KdV (\ref{eq1}) and higher order KdV (\ref{hierarchy}). We first introduce the 
polynomials of $\xi$ for $n=0, 1,
2, \dots$ \cite{dub, kri, Tian3}:
\begin{equation} \label{eq15} 
P_n(\xi, u_1, u_2, u_3) = \xi^{n+1} + a_{n, 1} \xi^n + \dots + a_{n,
  n+1} \ ,
\end{equation}
where the coefficients, $a_{n, 1}, a_{n, 2}, \dots, a_{n, n+1}$ are
uniquely determined by the two conditions
\begin{equation} \label{eq16}
\frac{ P_n(\xi, u_1, u_2, u_3) }{\sqrt{ (\xi - u_1)(\xi-u_2)(\xi-u_3)}
} = \xi^{n-1/2} + \mathcal{O}(\xi^{-3/2}) \quad  \mbox{for large $|\xi|$} 
\end{equation}
and
\begin{equation} \label{eq17}
\int_{u_3}^{u_2} \frac{ P_n(\xi, u_1, u_2, u_3)}{\sqrt{ (\xi- u_1)
    (\xi-u_2)(\xi-u_3)}} d\xi = 0 \ .
\end{equation}
Here the sign of the square root is given by $\sqrt{(\xi - u_1)(\xi - u_2)(\xi - u_3)} > 0$
for $\xi > u_1$ and the branch cuts are along $(- \infty, u_3)$ and $(u_2, u_1)$.

In particular, 
\begin{equation}
\label{P01}
P_0(\xi,u_1, u_2, u_3) = \xi + a_{0, 1} \ , \quad
P_1(\xi,u_1, u_2, u_3) = \xi^2 - \frac{1}{2} (u_1 + u_2 + u_3) \xi + a_{1, 2} \ ,
\end{equation}
where 
\begin{align*}
a_{0, 1} &= (u_1 - u_3) \frac{E(s)}{K(s)} - u_1 \ , \\
a_{1, 2} &= \frac{1}{3} ( u_1 u_2 + u_1 u_3 + u_2 u_3 ) + \frac{1}{6}
(u_1 + u_2 + u_3) a_{0, 1} \ . 
\end{align*}
Here 
\begin{equation*}
s = \frac{u_2 - u_3}{u_1 - u_3} 
\end{equation*}
and $K(s)$ and $E(s)$ are complete elliptic integrals of the first and
second kind.  

$K(s)$ and $E(s)$ have some well-known properties \cite{Tian1, Tian2}. They have the 
expansions
\begin{eqnarray}
K(s) & = & \frac{\pi}{2} [1 + \frac{s}{4} + \frac{9}{64} s^{2}
+ \cdots + (\frac{1 \cdot 3 \cdots (2n-1)}{2 \cdot 4 \cdots 2n})^{2} s^{n}
+ \cdots] \ , \label{K}\\
E(s) & = & \frac{\pi}{2} [1 - \frac{s}{4} - \frac{3}{64} s^{2}
 - \cdots - \frac{1}{2n-1}(\frac{1 \cdot 3 \cdots (2n-1)}{
2 \cdot 4 \cdots 2n})^{2} s^{n} -
\cdots] \ , \label{E}
\end{eqnarray}
for $|s| < 1$. They also have the asymptotics
\begin{eqnarray}
K(s) & \approx & \frac{1}{2} \log \frac{16}{1 - s}  \label{K2} \ , \\
E(s) & \approx & 1 + \frac{1}{4}(1 - s)[\log \frac{16}{1 - s} - 1] \ , \label{E2} 
\end{eqnarray}
as $s$ is close to $1$. Furthermore,
\begin{eqnarray}
\frac{d K(s)}{d s} & = & \frac{E(s) - (1-s)K(s)}{2s(1-s)} \ , \label{K3} \\
\frac{d E(s)}{d s} & = & \frac{E(s) - K(s)}{2s} \ . \label{E3} 
\end{eqnarray}
It immediately follows from (\ref{K}) and (\ref{E}) that
\begin{equation}
\frac{1}{1 - \frac{s}{2}} < \frac{K(s)}{E(s)} < \frac{1-\frac{s}{2}}{1-s}
\hspace*{.5in} for ~ 0 < s < 1 \ . \label{KE}
\end{equation}

The eigenspeeds of the Whitham equations (\ref{KdVW}) are defined in terms of
$P_0$ and $P_1$ of (\ref{P01}), 
\begin{equation*}  
\lambda_i(u_1, u_2, u_3) = 12 \frac{
  P_1(u_i, u_1, u_2, u_3) }{P_0(u_i, u_1, u_2, u_3)} \ , \quad i=1,2,3 \ ,
\end{equation*}
which give 
\begin{align}
\lambda_1(u_1, u_2, u_3) &= 2 (u_1 + u_2 + u_3) + 4(u_1 - u_2)
\frac{K(s)}{E(s)} \ , \nonumber \\
\lambda_2(u_1, u_2, u_3) &= 2(u_1 + u_2 + u_3) + 4 (u_2 - u_1)
\frac{ sK(s) }{E(s) - (1-s) K(s)} \ , \label{lambda} \\
\lambda_3(u_1, u_2, u_3) &= 2(u_1 + u_2 + u_3) + 4(u_2-u_3)
\frac{K(s)}{E(s) - K(s)} \ . \nonumber 
\end{align}

In view of (\ref{K}-\ref{E2}), we find that $\lambda_{1}$,
$\lambda_{2}$ and
$\lambda_{3}$ have behavior:

(1) At $u_{2}$ = $u_{3}$:
\begin{equation}
\label{tr}
\begin{array}{ll}
\lambda_{1}(u_1, u_2, u_3) = 6 u_{1} \ , \\
\lambda_{2}(u_1, u_2, u_3) =
\lambda_{3}(u_1, u_2, u_3) = 12 u_{3} - 6 u_{1} \ .
\end{array}
\end{equation}

(2) At $u_{1}$ = $u_{2}$:
\begin{equation}
\label{le}
\begin{array}{ll}
\lambda_{1}(u_1, u_2, u_3) =
\lambda_{2}(u_1, u_2, u_3) = 4 u_{1} + 2 u_{3} \ , \\
\lambda_{3}(u_1, u_2, u_3) = 6 u_{3} \ .
\end{array}
\end{equation}

The eigenspeeds of the Whitham equations (\ref{mKdVW})
are 
\begin{equation} \label{eq18a}
\mu_i^{(m)}(u_1, u_2, u_3) = 4^m (2m + 1) \frac{                                                                                                         
  P_m(u_i, u_1, u_2, u_3) }{P_0(u_i, u_1, u_2, u_3)} \ , \quad i=1,2,3 \ .                                                                                       
\end{equation}

The polynomial $4^m (2m + 1) P_m(\xi, u_1,u_2,u_3)$ can be expressed as \cite{gra}
\begin{equation}
\label{P}
4^m (2m + 1) P_m(\xi,u_1,u_2,u_3) = 2 (\xi - u_1)(\xi - u_2)(\xi - u_3) \Phi(\xi,u_1,u_2,u_3) + Q(\xi,u_1,u_2,u_3) \ .
\end{equation}

The function $\Phi(\xi,\vec{u})$
satisfies the boundary value problem for
the Euler-Poisson-Darboux equations
\begin{eqnarray}
2(u_i - u_j) {\p^2 \Phi \over \p u_i \p u_j} &=& {\p \Phi \over \p u_i} -
{\p \Phi \over \p u_j} \ , \label{ph1} \\
2(\xi - u_i) {\p^2 \Phi \over \p \xi \p u_i} &=& {\p \Phi \over \p \xi} -
2{\p \Phi \over \p u_i} \ , \label{ph2} \\
\Phi(u, u, u, u) &=& {2 \over 3} {d^2 \over d u^2}
[ C_m u^m ]
\ , \label{ph3}
\end{eqnarray}
where
\begin{equation}
\label{Cm}
C_m = {2^{2m+1} \over \int_0^1 {t^m \over \sqrt{1 - t}} d t} = {2^m (2m+1)!! \over m!} \ .
\end{equation}
The function $Q(\xi,\vec{u})$ is a quadratic polynomial in $\xi$;
\begin{eqnarray}
\label{Q}
Q(\xi, u_1,u_2,u_3) &=& 2 (\xi - u_2)(\xi - u_3) {\p q(u_1,u_2,u_3) \over \p u_1}
+ 2 (\xi - u_1)(\xi - u_3) {\p q(u_1,u_2,u_3) \over \p u_2}  \\
&& +2 (\xi - u_1)(\xi - u_2) {\p q(u_1,u_2,u_3) \over \p u_3}
+ q(u_1, u_2, u_3) P_0(\xi, u_1, u_2, u_3) \no
\end{eqnarray}
and $q(\vec{u})$ is the solution of the boundary value problem for
another version of the Euler-Poisson-Darboux equations
\begin{eqnarray}
2(u_i - u_j) {\p^2 q \over \p u_i \p u_j} &=& {\p q \over \p u_i} -
{\p q \over \p u_j} \ , \quad \mbox{$i, j= 1, 2, 3$} \ , \label{q1} \\
q(u, u, u) &=& C_m u^m \ .\label{q2}
\end{eqnarray}

The solution of equations (\ref{ph1}-\ref{ph3}) and that of (\ref{q1}) and 
(\ref{q2}) can be solved explicitly \cite{gra}.
In particular, the solution of (\ref{q1}) and (\ref{q2}) is \cite{Tian1} 
\begin{equation}
\label{iq}
q(u_1, u_2, u_3) = {C_m \over 2 \sqrt{2} \pi} \int_{-1}^1 \int_{-1}^1 {({1 + \mu \over 2}{1 + \nu \over 2} u_1
+ {1 + \mu \over 2}{1 - \nu \over 2} u_2 + {1 -\mu \over 2} u_3)^m \over \sqrt{(1 - \mu)(1 - \nu^2)}} \ d \mu d \nu    
\ .
\end{equation}

The speeds $\mu_i^{(m)}$'s of (\ref{eq18a}) for $m>1$ are connected to $\mu_i^{(m)}$'s for $m=1$, which are also given 
by (\ref{lambda}). 
\begin{lem}
For $i= 1, 2, 3$,
\begin{equation}
\label{q}
\mu_i^{(m)}(u_1, u_2, u_3) = {1 \over 2} [ \lambda_i(u_1, u_2, u_3) - 2 (u_1 + u_2 + u_3)] \ {\p q(u_1, u_2, u_3) 
\over \p u_i} + q(u_1, u_2, u_3) \ .
\end{equation}
\end{lem}
\begin{proof}
We use (\ref{eq18a}), (\ref{P}) and (\ref{Q}) to write
\begin{eqnarray}
\mu_1^{(m)}(u_1, u_2, u_3) &=& {Q(u_1, u_1, u_2, u_3) \over P_0(u_1, u_1, u_2, u_3)}   \no \\
&=& { 2(u_1 - u_2)(u_1 - u_3) \over P_0(u_1, u_1, u_2, u_3)} \ {\p q(u_1, u_2, u_3) \over \p u_i}+q(u_1, u_2, u_3) \ .
\label{mu1}
\end{eqnarray}
In particular, when $m=1$, since the corresponding $q = 2(u_1+u_2+u_3)$, we obtain
$$\lambda_1(u_1, u_2, u_3) = { 4(u_1 - u_2)(u_1 - u_3) \over P(u_1, u_1, u_2, u_3)} + 2(u_1+u_2+u_3) \ .$$
This together with (\ref{mu1}) proves formula (\ref{q}) for $i=1$. The cases for $i=2, 3$ can be shown in the same
way.
\end{proof}

\begin{lem} \cite{Tian3} 
\begin{enumerate}[1.]

\item
\begin{equation}
\label{a23}
{\p \mu_i^{(m)} \over \p u_j} = {{\p \lambda_i \over \p u_j} \over \lambda_i - \lambda_j} \ [\mu_i^{(m)} - \mu_j^{(m)}] \ ,       
\quad \quad i, j = 1, 2, 3; i \neq j \ .
\end{equation}

\item
\begin{equation}
\label{Pm0}
{\p \over \p u_i} \left( {P_m(\xi, u_1, u_2, u_3) \over \sqrt{(\xi - u_1)(\xi - u_2)(\xi - u_3)}} \right)
= {\mu_i^{(m)}(u_1, u_2, u_3) \over 4^m (2m +1)} {\p \over \p u_i} \left( {P_0(\xi, u_1, u_2, u_3) \over 
\sqrt{(\xi - u_1)(\xi - u_2)(\xi - u_3)}} \right)
\end{equation}
for $i=1,2,3$ and $\eta \neq u_1, u_2, u_3$.
\end{enumerate}

\end{lem}

The following calculations are useful in the subsequent sections.

Using formula (\ref{q}) for $\mu_2$ and $\mu_3$ and formulae (\ref{lambda}) for
$\lambda_2$ and $\lambda_3$, we obtain
\begin{equation}
\label{M}
\mu_2^{(m)}(u_1, u_2, u_3) - \mu_3^{(m)}(u_1, u_2, u_3) = {2(u_2 - u_3) K \over (K-E)[E-(1-s)K]}
M(u_1, u_2, u_3) \ ,
\end{equation}
where
\begin{equation}
\label{M1}
M(u_1, u_2, u_3) = [{\p q \over \p u_3} + (1-s) {\p q \over \p u_2}] E -(1-s)({\p q \over \p u_2} +
{\p q \over \p u_3}) K \ . 
\end{equation}
We then use (\ref{K3}), (\ref{E3}) and (\ref{q1}) to calculate
\begin{eqnarray}
{\p M(u_1,u_2,u_3) \over \p u_2} &=& {1 \over 2} \ {p_1(u_1, u_2, u_3) \over u_1 - u_3} [E - K] \ , \label{pM} \\
{\p M(u_1,u_2,u_3) \over \p u_3} &=& {1 \over 2} \ {p_2(u_1, u_2, u_3) \over u_1 - u_3} [E - (1-s) K] + 
{3 \over 2} \ {M(u_1, u_2, u_3) \over u_1 - u_3} \ , \label{pM'}
\end{eqnarray}
where   
\begin{equation}
\label{poly}
p_1(u_1, u_2, u_3) = 2(u_1 - u_2) {\p \over \p u_2} \mbox{div}(q) - \mbox{div}(q) \ , ~~
p_2(u_1, u_2, u_3) = 2(u_1 - u_3) {\p \over \p u_3} \mbox{div}(q) - \mbox{div}( q) \ .
\end{equation}

We next consider
\begin{equation}
\label{F}
F(u_1, u_2,u_3) := {\mu_2(u_1, u_2, u_3)-\mu_3(u_1, u_2, u_3) \over u_2 - u_3} \ .
\end{equation}
Using formula (\ref{q}) for $\mu_2$ and $\mu_3$ and formulae (\ref{lambda}) for
$\lambda_2$ and $\lambda_3$, we obtain
\begin{eqnarray*}
F &=& - 2 {(1-s)K \over E - (1-s)K}
{\p q \over \p u_2} +
2 {K \over K - E} {\p q \over \p u_3} \\
&=& - 4 {s(1-s)K \over E - (1-s)K} (u_1 - u_3) {\p^2 q \over \p u_2 \p u_3}
+ 2 [ {K \over K - E} - {(1-s)K \over E - (1-s)K}] {\p q \over \p u_3} \ ,
\end{eqnarray*}
where we have used equations (\ref{q1}) in the last
equality. Finally, we use the expansions (\ref{K}-\ref{E}) for $K$ and $E$ to obtain
\begin{equation}
F(u_1, u_2,u_3) = -4 [ (2 - {7 \over 4} s + \cdots ) (u_1 - u_3)
{\p^2 q \over \p u_2 \p u_3} + (-{3 \over 4} + O(s^2)){\p q \over \p u_3}
] \ . \label{F2}
\end{equation}

\section{Self-similar Solutions}

In this section, we construct the self-similar solution of the Whitham equations
(\ref{mKdVW}) when $m \geq 2$ for the initial function (\ref{step}).
The $m=2$ result has already been obtained in \cite{PT}. Even in the $m=2$ case, the
key calculations presented here are different from those in \cite{PT}.

\begin{thm}(see Figure 1.)
\label{main3}
For the step-like initial data $u_0(x)$ of (\ref{step}), the solution of the Whitham equations
(\ref{mKdVW}) is given by
\begin{equation}
\label{ws1}
u_1 = 1 \ , \quad x = \mu_2^{(m)}(1, u_2, u_3) \ t \ , \quad x = \mu_3^{(m)}(1, u_2, u_3) \ t
\end{equation}
for $\alpha t < x \leq \beta t$ and by
\begin{equation}
\label{ws2}
u_1 = 1 \ , \quad x = \mu_2^{(m)}(1, u_2, 0) \ t \ , \quad u_3 = 0
\end{equation}
for $\beta t \leq x < \gamma t$, where $\alpha = \mu_2^{(m)}(1, u^*, u^*)$, $\beta = \mu_2^{(m)}(1, u^{**}, 0)$
and $\gamma = \mu_2^{(m)}(1, 1, 0) = q(1,1,0) = 4^m$. Here, $u^*$ is uniquely determined by the equation
\begin{equation}
\label{p1} p_1(1, u^*, u^*) = 0 \ , 
\end{equation}
and $u^{**}$ is uniquely given by the equation
\begin{equation}
\label{mid}
\mu_2^{(m)}(1, u^{**}, 0) - \mu_3^{(m)}(1, u^{**}, 0) = 0 \ .
\end{equation}
Outside the region $\alpha t < x < 4^m t$, the solution of the Burgers
type equation (\ref{mBurgers}) is given by
\begin{equation}
\label{bs1}  u \equiv 1 \quad \mbox{$x \leq \alpha t$}
\end{equation}
and
\begin{equation}
\label{bs2}  u \equiv 0 \quad \mbox{$x \geq 4^m t$} \ .
\end{equation}
\end{thm}

\begin{figure}[h] \label{fig1}
\begin{center}
\resizebox{15cm}{5cm}{\includegraphics{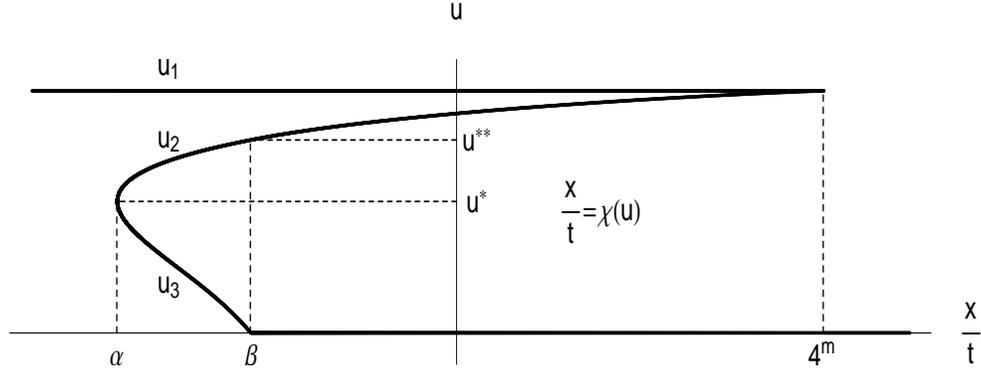}}
\caption{Self-Similar solution of the Whitham equations for $m \geq 2$. The curve defines the function $\chi(u)$.}
\end{center}
\end{figure}

The boundaries $x = \alpha t$ and $x = 4^m t$ are called the trailing and leading edges, respectively.
They separate the solutions of the Whitham equations and Burgers type equations.
The Whitham solution matches the Burgers type solution in the following fashion (see Figure 1.):
\begin{eqnarray}
\label{tr1}
u_1 &=& \mbox{the Burgers type solution defined outside the region} \ , \\
\label{tr2}
u_2 &=& u_3 \ ,
\end{eqnarray}
at the trailing edge;
\begin{eqnarray}
\label{le1}
u_1 &=& u_2 \ , \\
\label{le2}
u_3 &=& \mbox{the Burgers type solution defined outside the region} \ ,
\end{eqnarray}
at the leading edge.

The proof of Theorem \ref{main3} is based on a series of lemmas.

We first show that the solution defined by either formulae (\ref{ws1}) or (\ref{ws2})
indeed satisfies the Whitham equations (\ref{mKdVW}) \cite{dub, PT, tsa}.

\begin{lem}
\label{tsa}
\begin{enumerate}
\item The functions $u_1$, $u_2$ and $u_3$ determined by equations (\ref{ws1})
give a solution of the Whitham equations (\ref{mKdVW}) as long as $u_2$ and $u_3$
can be solved from (\ref{ws1}) as functions of $x$ and $t$.

\item The functions $u_1$, $u_2$ and $u_3$ determined by equations (\ref{ws2})
give a solution of the Whitham equations (\ref{mKdVW}) as long as $u_2$
can be solved from (\ref{ws2}) as a function of $x$ and $t$.
\end{enumerate}

\end{lem}

\begin{proof}

(1) $u_1$ obviously satisfies the first equation of (\ref{mKdVW}). To verify the
second and third equations, we observe that
\begin{equation}
\label{dia}
\frac{\p \mu_2^{(m)} }{\p u_3} = \frac{\p \mu_3^{(m)} }{\p u_2} = 0
\end{equation}
on the solution of (\ref{ws1}). To see this, we use (\ref{a23}) to calculate
$$\frac{\p \mu_2^{(m)} }{\p u_3} = {{\p \lambda_2 \over \p u_3} \over \lambda_2 - \lambda_3}
\ (\mu_2^{(m)} - \mu_3^{(m)}) = 0 \ .$$
The second part of (\ref{dia}) can be shown in the same way.

We then calculate the partial derivatives of the second equation of (\ref{ws1})
with respect to $x$ and $t$.
$$ 1 = \frac{\p \mu_2^{(m)} }{\p u_2} \ t u_{2x} \ , \quad 0 = \frac{\p \mu_2^{(m)} }{\p u_2} 
\ t u_{2t} + \mu_2^{(m)} \ ,$$
which give the second equation of (\ref{mKdVW}).

The third equation of (\ref{mKdVW}) can be verified in the same way.

(2) The second part of Lemma 3.2 can easily be proved.

\end{proof}

We now determine the trailing edge. Eliminating $x$ and $t$ from the last two equations of (\ref{ws1})
yields
\begin{equation}
\label{m23}
\mu_2^{(m)}(1, u_2, u_3) - \mu_3^{(m)}(1, u_2, u_3) = 0 \ .
\end{equation}
Since it degenerates at $u_2 = u_3$, we replace (\ref{m23}) by
\begin{equation}
\label{F1}
F(1, u_2,u_3) := {\mu_2^{(m)}(1, u_2, u_3)-\mu_3^{(m)}(1, u_2, u_3) \over u_2 - u_3} = 0 \ .
\end{equation}
Here, the function $F$ is also defined in (\ref{F}).

Therefore, at the trailing edge where $u_2=u_3$, i.e., $s=0$, equation
(\ref{F1}), in view of the expansion (\ref{F2}), becomes
$$8(1 - u_2){\p^2 q(1, u_2, u_2) \over \p u_2 \p u_3} - 3 {\p q(1, u_2, u_2) \over \p u_3} = 0 \ .$$
Since ${\p q \over \p u_2}= {\p q \over \p u_3}$ and ${\p^2 q \over \p u_2^2}= 
3{\p^2 q \over \p u_2 \p u_3}$
on $u_2=u_3$ because of (\ref{q1}), this equation is exactly equation (\ref{u*}). 

\begin{lem}
\label{u*}
Equation $p_1(1, \xi, \xi) = 0$ has a simple zero, denoted by $u^*$, in the region $0 < \xi < 1$, counting multiplicities.
Furthermore, $p_1(1, \xi, \xi)$ is positive when $\xi < u^*$ and negative when $\xi > u^*$.
\end{lem}

\begin{proof}
We first simplify the polynomial $p_1$ of (\ref{poly}).
In view of formula (\ref{iq}) for $q$, we use the fact that $q$ is symmetric in $u_1$, $u_2$ and $u_3$ to
obtain 
\begin{equation}
\label{divq}
\left( {\p q(u_1, u_2, u_3) \over \p u_1} + {\p q(u_1, u_2, u_3) \over \p u_2}
+ {\p q(u_1, u_2, u_3) \over \p u_3} \right)|_{u_2=u_3=\xi} = C_m U_0(\xi, u_1) \ ,
\end{equation}
where 
\begin{equation}
\label{fU}
U_0(\xi, u_1) = {1 \over 2 \sqrt{2}} \int_{-1}^{1} {m({1 + \mu \over 2}
\xi + {1- \mu \over 2} u_1) \over \sqrt{1 - \mu} }^{m-1} d \mu \ .
\end{equation}
We can then write $p_1$ as
\begin{equation}
\label{U}
p_1(1, \xi, \xi) = C_m [(1 - \xi) {\p U_0(\xi, 1) \over \p \xi} - U_0(\xi, 1)] \ .
\end{equation}

Denoting the function in the parenthesis of (\ref{U}) by $p(\xi)$, we claim that
\begin{equation}
\label{p}
{d^k p(0) \over d \xi^k} > 0 \ , ~~~~~~~~ {d^k p(1) \over d \xi^k} < 0
\end{equation}
for $k = 0, 1, 2, \cdots, m-2$.

Obviously,
$${d^k p(\xi) \over d \xi^k} = (1 - \xi) {d^{k+1} U_0(\xi,1) \over d \xi^{k+1}} -
(k+1) {d^k U_0(\xi,1) \over d \xi^k} \ .$$
Since, ${d^k U_0(\xi,1) \over d \xi^k}$ is a positive function, this proves the second
inequality of (\ref{p}).

To prove the first inequality of (\ref{p}),
we use formula (\ref{fU}) to calculate
\begin{equation*}
\dfrac{d^kU_0(\xi,1)}{d\xi^k}|_{\xi=0} = {m(m-1) \cdots (m-k) \over 2^{m + {1 \over 2}}}
\int_{-1}^1 (1 - \mu)^{m - k - {3 \over 2}} (1 + \mu)^k d \mu \ . 
\end{equation*}
The integral on the right can be evaluated using an iteration formula. Denote this integral
by $A_{m,k}$. An integration by parts gives $A_{m,k} = {2k \over 2m - 2k -1} \ A_{m,k-1}$.
Since $A_{m,0} = 2^{m+{1\over2}}/(2m-1)$, we thus obtain
$$A_{m,k} = {2^{m+k +{1 \over 2}} k! \over (2m-1)(2m-3) \cdots (2m-2k-1)} \ ,$$
which gives
\begin{equation}
\label{Uk}
\dfrac{d^kU_0(\xi,1)}{d\xi^k}|_{\xi=0} = {2^k k! m(m-1) \cdots (m-k) \over (2m-1)(2m-3) 
\cdots (2m-2k-1)} \ .
\end{equation}
Therefore
\[
\dfrac{d^kp(0)}{d\xi^k}= {d^{k+1} U_0(\xi,1) \over d \xi^{k+1}}|_{\xi =0} -
(k+1) {d^k U_0(\xi,1) \over d \xi^k}|_{\xi =0}
= {2^k (k+1)! m(m-1) \cdots (m-k) \over (2m-1)(2m-3) \cdots (2m-2k-3)} > 0 \ . 
\]

We now use (\ref{p}) to prove the existence and uniqueness of the zero of function
$p^{(k)}(\xi)$, $k=0, 1, \cdots, (m-2)$. First, it follows from (\ref{p}) that
$p^{(k)}(\xi)$ has an odd number of zeros in $0 < \xi < 1$,
counting multiplicities.
Second, if $p^{(k)}(\xi)$ has more than one zero, it must have at least three zeros.
Consequently, $p^{(k+1)}(\xi)$ will have more than one zero; so it must also have at least three zeros.
Repeating this argument, we see that $p^{(m-2)}$ must have at least three zeros.
This is an impossibility since
$p(\xi)$ is a polynomial of degree $m-1$ because $U_0(\xi, 1)$ is so.
Therefore, $p^{(k)}(\xi)$ has one and only one zero for $\xi \in (0, 1)$ when $k = 0, 1, \cdots, (m-2)$.
In particular, the $k=0$ case proves Lemma \ref{u*}.
\end{proof}

\begin{lem}
\label{trailing}
Equation (\ref{F1}) has a unique solution satisfying $u_2=u_3$. The solution
is $u_2=u_3=u^*$. The rest of equations (\ref{ws1}) at the trailing edge
are $u_1=1$ and
$x/t = \mu_2^{(m)}(1, u^*, u^*)$.
\end{lem}

Having located the trailing edge, we now solve equations (\ref{ws1}) in the
neighborhood of the trailing edge. We first consider equation (\ref{F1}).
We use (\ref{F2}) to differentiate $F$ at the trailing edge $u_1=1$, $u_2=u_3=u^*$
\begin{eqnarray}
{\p F(1, u^*, u^*) \over \p u_2} =
{\p F(1, u^*, u^*) \over \p u_3} &=& 10 {\p^2 q(1, u^*, u^*) \over \p u_2 \p u_3} -
8 (1 - u^*) {\p^3 q(1, u^*, u^*) \over \p u_2^2 \p u_3} > 0 \no \\
&=& - {C_m \over 2} {\p \over \p \xi} [(1-\xi) {\p U_0(\xi, 1) \over \p \xi} - U_0(\xi, 1)]_{\xi = u^*} > 0 
\label{u+}
\ ,
\end{eqnarray}
where in the second equality we have used (\ref{q1}), (\ref{divq}) and 
identities ${\p^2 q \over \p u_2^2}= 
{\p^2 q \over \p u_3^2} = 3{\p^2 q \over \p u_2 \p u_3}$ on $u_2=u_3$.
The inequality is a consequence of Lemma \ref{u*}.
  
Inequality (\ref{u+}) shows that equation (\ref{F1}) or equivalently (\ref{m23}) can be
inverted to give $u_2$ as a decreasing
function of $u_3$
\begin{equation}
\label{B} u_2 = B(u_3)
\end{equation}
in a neighborhood of $u_2=u_3= u^*$.

We will extend the solution (\ref{B}) of equation (\ref{m23}) by decreasing $u_3$ in the region
$0 < u_3 < u^* < u_2 < 1$ as far as possible.
We need to evaluate the derivatives
${\p (\mu_2^{(m)} - \mu_3^{(m)}) \over \p u_2}$ and ${\p (\mu_2^{(m)} - \mu_3^{(m)}) \over \p u_3}$ on the solution of (\ref{ws1}).
It follows from (\ref{M}), (\ref{pM}) and (\ref{pM'}) that  
\begin{eqnarray}
{\p [\mu_2^{(m)} - \mu_3^{(m)}] \over \p u_2}
&=& - {s K p_1(1, u_2, u_3) \over E - (1-s)K} \label{u22} \ ,\\
{\p [\mu_2^{(m)} - \mu_3^{(m)}] \over \p u_3}
&=& {s K p_2(1, u_2, u_3) \over K - E} \label{u33} \ ,
\end{eqnarray}
on the solution of (\ref{ws1}). 

We first study the two polynomials $p_1$ and $p_2$ of (\ref{poly}).

\begin{lem}
\label{p12}
For each $0 \leq u_3 < 1$, the polynomial $p_1(1, u_2, u_3)$, as a function of $u_2$, has only one zero
in the region $0 < u_2 < 1$, counting multiplicities. Furthermore, $p_1(1, u_2, u_3)$ is positive
when $u_2$ is on the left of this zero and negative when $u_2$ is on the right.

For each $0 \leq u_2 < 1$, the polynomial $p_2(1, u_2, u_3)$, as a function of $u_3$, has only one zero
in the region $0 < u_3 < 1$, counting multiplicities. Furthermore, $p_1(1, u_2, u_3)$ is positive
when $u_3$ is on the left of this zero and negative when $u_3$ is on the right.
\end{lem}

\begin{proof}
We will prove the first part of the lemma; the second part follows from $p_2(1,u_2, u_3) = p_1(1, u_3, u_2)$.

The proof of the first part is similar to the proof of Lemma \ref{u*}. We will go through it 
briefly. 

We first have
$${\p^k p_1(1, 0, u_3) \over \p u_2^k} > 0 \ , \quad {\p^k p_1(1, 1, u_3) \over \p u_2^k} < 0 $$
for $k=0, 1, 2, \cdots, m-2$.    
The second inequality immediately follows from formula (\ref{poly}) for $p_1$. The first inequality is
derived from a formula similar to (\ref{Uk}).

The rest of the proof is the same as the proof of Lemma \ref{u*}.
\end{proof}

We now continue to extend the solution (\ref{B}) of equation (\ref{m23}) in the region
$0 < u_3 < u^* < u_2 < 1$ as far as possible. When $u_2$ and $u_3$ are close to $u^*$,
because of (\ref{u+}), we have ${\p (\mu_2^{(m)} - \mu_3^{(m)}) \over \p u_2} > 0$ and
${\p (\mu_2^{(m)} - \mu_3^{(m)}) \over \p u_3} >  0$ on the solution of (\ref{m23}). These along with
(\ref{u22}) and (\ref{u33}) show that
\begin{equation}
\label{inq}
p_1(1, u_2, u_3) < 0 \ , ~~ p_2(1, u_2, u_3) > 0
\end{equation}
when $u_2$ and $u_3$ are close to $u^*$.

\begin{lem}
\label{p12>}
Inequalities (\ref{inq}) hold on the solution of equation (\ref{m23})
as long as $0 < u_3 < u^* < u_2 < 1$. 
\end{lem}

\begin{proof}

We first prove the first inequality of (\ref{inq}). In view of (\ref{M}), the solution of equation 
(\ref{m23}) is also governed by an equivalent equation $M(1, u_2, u_3) = 0$. For each $0 < u_3 < u^*$,
we study the zero of the $u_2$-variable function $M(1, u_2, u_3)$ in the interval $u_3 < u_2 < 1$. 
We obtain from formula (\ref{M1}) for $M$ that $M(1, u_3, u_3) = 0$ and $M(1, 1, u_3) > 0$. 
$M(1, u_2, u_3)$ is a decreasing function of $u_2$ when $u_2$ is on the immediate right of $u_3$. To see this,
we note that $p_1(1, u_3, u_3) > 0$ for $u_3 < u^*$ according to Lemma \ref{u*}. This and 
(\ref{pM}) prove that $M$ is decreasing for $u_2$ on the immediate right of $u_3$ 
because $E - K < 0$ for $s > 0$;
so $M(1, u_2, u_3) < 0$ for such $u_2$. Therefore,
$M(1, u_2, u_3)$ has a $u_2$-zero in the interval $u_3 < u_2 < 1$ when $0 < u_3 < u^*$. Because of the
uniqueness of the $u_2$-zero of $p(1, u_2, u_3)$ according to Lemma \ref{p12}, we conclude that
$M(1, u_2, u_3)$ has only one zero and that this zero is on the right of the zero of $p_1$. Hence,
the zero of $M$ is exactly given by $u_2 = B(u_3)$ and $p_1$ is negative on the solution. This proves
the first inequality of (\ref{inq}).

We now prove the second inequality of (\ref{inq}) by contradiction.
Suppose it first fails at $\bar{u}_2$ and $\bar{u}_3$, where
$0 < \bar{u}_3 < u^* < \bar{u}_2 < 1$; i.e., 
\begin{equation}
p_2(1, B(u_3), u_3) >  0 \ \mbox{when $\bar{u}_3 < u_3 < u^*$} \ , \quad 
p_2(1, B(u_3), u_3) =0 \ \mbox{when $u_3 = \bar{u}_3$} \ . \label{2'} 
\end{equation}
It then follows from (\ref{u22}) and (\ref{u33}) that solution (\ref{B}) of equation  
(\ref{m23}) has a zero derivative at $u_3 = \bar{u}_3$; i.e.,
$$B'(\bar{u}_3) = 0 \ .$$
Hence, 
$${d \over d u_3} p_2(1, B(u_3), u_3) |_{u_3 = \bar{u}_3} = {\p \over \p u_2}p_2(1, B(\bar{u}_3),
\bar{u}_3) B'(\bar{u}_3) + {\p \over \p u_3}p_2(1, B(\bar{u}_3), \bar{u}_3) < 0 \ ,$$
where the first term vanishes because of $B'(\bar{u}_3) = 0$ and the second term is negative
according to Lemma \ref{p12}. In view of $p_2(1, B(\bar{u}_3), \bar{u}_3)=0$, this implies
that $p_2(1, B(u_3), u_3) < 0$ when $u_3$ is on the immediately right of $\bar{u}_3$. That
contradicts (\ref{2'}). This proves the second inequality of (\ref{inq}).

\end{proof}

It follows from (\ref{u22}), (\ref{u33}) and Lemma 3.6 that
\begin{equation}
\label{dia2}
{\p [\mu_2^{(m)} - \mu_3^{(m)}] \over \p u_2} > 0 \ , \quad {\p [\mu_2^{(m)} - \mu_3^{(m)}] \over \p u_3} >  0
\end{equation}
on the solution of (\ref{m23}).
Solution (\ref{B}) of equation (\ref{m23})
can then be extended as a decreasing function of $u_3$ as long as $0 < u_3 < u^* < u_2 < 1$.

There are two possibilities: (1) $u_2$ touches $1$ before or simultaneously
as $u_3$ reaches $0$ and (2) $u_3$ touches $0$ before $u_2$ reaches $1$.

It follows from (\ref{le}) and (\ref{q}) that
$$\mu_2^{(m)}(1,1,u_3) > \mu_3^{(m)}(1,1,u_3) \quad \mbox{for $0 \leq u_3 < 1$} \ .$$
This shows that (1) is impossible. Hence, $u_3$ will touch $0$ before $u_2$
reaches $1$. When this happens, equation (\ref{m23}) becomes
equation (\ref{mid}).

\begin{lem}
\label{u2*}
Equation (\ref{mid}) has a simple zero in the region $0 < u_2 < 1$, counting
multiplicities. Denoting the zero by $u^{**}$, then $\mu_2^{(m)}(1, u_2, 0)-\mu_3^{(m)}(1, u_2, 0)$ is positive
for $u_2 > u^{**}$ and negative for $u_2 < u^{**}$. Furthermore, $0 < u^{***} < u^{**} < 1$ where $u^{***}$
is the unique zero of $p_1(1, u_2, 0)$.
\end{lem}

\begin{proof}
We use (\ref{M}) and (\ref{pM}) to prove the lemma.
In equation (\ref{M}), $K-E$ and $E-(1-s)K$ are all positive for $0<s<1$ in view of (\ref{KE}).

Denoting the unique zero of $p_1(1, u_2, 0)$ by $u^{***}$, it then follows from Lemma \ref{p12} that
$p_1(1, u_2, 0) > 0$ when $0 < u_2 < u^{***}$ and $p_1(1, u_2, 0) < 0$ when $u^{***} < u_2 < 1$.
Since $M(1, u_2, 0)$ of (\ref{M1}) vanishes at $u_2=0$ and is positive at $u_2=1$ in view of (\ref{K}-\ref{E2}),
we conclude from the derivative (\ref{pM}) that $M(1, u_2, 0)$ has a simple zero in $0<u_2<1$. 
This zero is exactly
$u^{**}$ and the rest of the theorem can be proved easily.
\end{proof}

Having solved equation (\ref{m23}) for $u_2$ as a decreasing function of $u_3$
for $0 < u_3 < u^*$, we turn to equations (\ref{ws1}). Because of (\ref{dia}) and (\ref{dia2}),
the third equation of (\ref{ws1}) gives $u_3$ as a decreasing function of $x/t$ for
$\alpha \leq x/t \leq \beta$, where $\alpha = \mu_2^{(m)}(1, u^*, u^*)$ and 
$\beta = \mu_2^{(m)}(1, u^{**}, 0)$.
Consequently, $u_2$ is an increasing function of $x/t$ in the same interval.

\begin{lem}
\label{Sws1}
The last two equations of (\ref{ws1}) can be inverted to give $u_2$ and $u_3$ as
increasing and decreasing functions, respectively, of the self-similarity variable
$x/t$ in the interval $\alpha \leq x/t \leq \beta$.
\end{lem}

We now turn to equations (\ref{ws2}). 
We first use (\ref{K3}), (\ref{E3}) and (\ref{lambda}) to calculate
the derivative of $\mu_2^{(m)}$ of (\ref{q})
\begin{eqnarray}
{\p \mu_2^{(m)} \over \p u_2} &=&  {1 \over 2} [\lambda_2 - 2(1 + u_2 + u_3)] {\p^2 q \over
\p u_2^2} + {1 \over 2} {\p \lambda_2 \over \p u_2} {\p q \over \p u_2} \no \\
&=& {2(u_2 - 1) s K \over E - (1-s) K} \ {\p^2 q \over \p u_2^2} + [{2 s K \over E - (1-s) K} 
+ 1 - {E^2 - (1-s)K^2 \over (E - (1-s)K)^2}] {\p q \over \p u_2} \no \\ 
& > & - {2s K \over E - (1-s)K} \ [ (1 - u_2) {\p^2 q \over \p u_2^2} - 
{\p q \over \p u_2}] \ ,  \label{*} 
\end{eqnarray}
where in the inequality we have used $(E - (1-s)K)^2 > E^2 - (1-s)K^2$, which is a consequence
of (\ref{KE}).

The polynomial in the parenthesis of (\ref{*}) is connected to $p_1(1,u_2,0)$;
indeed, 
\begin{equation}
\label{**}
p_1(1,u_2,0) = {2m + 1 \over m} \ [(1 - u_2) {\p^2 q(1,u_2,0) \over \p u_2^2} - 
{\p q(1,u_2,0) \over \p u_2}] \ .
\end{equation}
This follows from the identity
\begin{equation}
\label{identity3.27} 
2 m {\p q(1,u_2,0) \over \p u_3}  = (1 - u_2) {\p q(1,u_2,0) \over \p u_2}
+ m q(1,u_2,0) \ .\end{equation}
To see this, taking the derivative of (\ref{identity3.27}) and using formula (\ref{poly}) 
for $p_1$ and
equations (\ref{q1}) for $q$ yield (\ref{**}).

To prove (\ref{identity3.27}), we use the integral formula (\ref{iq}) for $q$ to calculate
both sides of the identity. The left equals
\begin{equation*}
\frac{C_m m^2}{4^{m-1} 2^{3/2} \pi} \left[ \int_{-1}^1
  (1+\mu)^{m-1} (1-\mu)^{1/2} d\mu \right] \left[ \int_{-1}^1 \left(
  (1+\nu) + (1-\nu) u_2 \right)^{m-1} (1-\nu^2)^{-1/2} d\nu \right] \,.
\end{equation*}
The right is
\begin{equation*}
\frac{C_m m}{4^{m-1} 2^{5/2} \pi} \left[ \int_{-1}^1 (1+\mu)^m
 (1-\mu)^{-1/2} d\mu \right]  
 \left[ \int_{-1}^1 \left( (1+\nu) +
 (1-\nu) u_2 \right)^{m-1} (1-\nu^2)^{-1/2} d\nu \right]  \ .
\end{equation*}
Both sides are equal in view of an easy identity
$$ m \int_{-1}^1 (1+\mu)^{m-1} (1-\mu)^{1/2} d\mu = \frac{1}{2}
\int_{-1}^1 (1+\mu)^m (1-\mu)^{-1/2} d\mu \ . $$
We have therefore proved identity (\ref{identity3.27}).


By Lemma \ref{p12}, $p_1(1,u_2,0)$ is negative for $u_2 > u^{***}$, where $u^{***}$
is the unique zero of $p_1$. Since $u^{***} < u^{**}$ according to Lemma \ref{u2*},
we conclude from (\ref{*}) and (\ref{**}) that ${\p \mu_2^{(m)} \over \p u_2} > 0$
on the solution of (\ref{ws2}) when $u_2 > u^{**}$. 
Hence, the second equation of (\ref{ws2}) can be solved for $u_2$ as an increasing
function of $x/t$ as long as $u^{**} < u_2 < 1$. When $u_2$ reaches $1$, we have
$$x/t = \gamma = \mu_2^{(m)}(1, 1, 0) \ .$$
We have therefore proved the following result.

\begin{lem}
\label{Sws2}
The second equation of (\ref{ws2}) can be inverted to give $u_2$ as an increasing
function of $x/t$ in the interval $\beta \leq x/t \leq \gamma$.
\end{lem}

We are ready to conclude the proof of Theorem 3.1.

The Burgers type solutions (\ref{bs1}) and (\ref{bs2}) are trivial.

According to Lemma \ref{Sws1}, the last two equations of (\ref{ws1}) determine $u_2$
and $u_3$ as functions of
$x/t$ in the region $\alpha \leq x/t \leq \beta$. By the first part of Lemma \ref{tsa}, the
resulting $u_1$, $u_2$ and $u_3$ satisfy the Whitham equations (\ref{mKdVW}).
Furthermore, the boundary conditions (\ref{tr1}) and (\ref{tr2}) are satisfied
at the trailing edge $x = \alpha \ t$.

Similarly, by Lemma \ref{Sws2}, the second equation of (\ref{ws2}) determines $u_2$
as a function of $x/t$ in the region $\beta \leq x/t \leq \gamma=4^m$. It then follows from
the second part of Lemma \ref{tsa} that $u_1$, $u_2$ and $u_3$ of (\ref{ws2}) satisfy
the Whitham equations (\ref{mKdVW}).
They also satisfy the boundary conditions (\ref{le1}) and (\ref{le2}) at the
leading edge $x = \gamma \ t$.

We have therefore completed the proof of Theorem \ref{main3}.


\section{The Minimization Problem}

The zero dispersion limit of the solution of the higher order KdV
equation (\ref{hierarchy}) with step-like initial function (\ref{step})
is also determined by
a minimization problem with constraints \cite{lax, lax2, ven}
\begin{equation}
\label{mini}
\underset{\{\psi \geq 0, \  \psi \in L^1 \}}
{\rm
Minimize}
\{ - \frac{1}{2 \pi} \int_0^1 \int_0^1 \log \Big|\frac{\eta - \mu}
{\eta + \mu
}\Big|
\psi(\eta) \psi(\mu) d \eta d \mu + \int_0^1 [\eta x - 4^m \eta^{2m+1} t]
\psi(\eta) d \eta \} \ .
\end{equation}

In this section, we will use the self-similar solution of Section 3 to construct
the minimizer for $m \geq 2$. The $m=2$ result has already been obtained in \cite{PT}. 
Even in the $m=2$ case, the
key calculations presented here are different from those in \cite{PT}.

We first define a linear operator
$$L \psi(\eta) = {1 \over 2 \pi} \int_0^1 log \left ( {\eta - \mu \over \eta + \mu} \right )^2
\psi(\mu) d \mu \ .$$
The variational conditions are
\begin{eqnarray}
L \psi = x \eta - 4^m t \eta^{2m+1} \quad \mbox{where $\psi > 0$} \ ,  \label{con1} \\
L \psi \leq x \eta - 4^m t \eta^{2m+1} \quad \mbox{where $\psi = 0$}  \ . \label{con2}
\end{eqnarray}
The constraint for the minimization problem is
\begin{equation}
\label{con}
\psi \geq 0 \ .
\end{equation}

The minimizer of (\ref{mini}) is given explicitly:

\begin{thm}
The minimizer of the variational problem (\ref{mini}) is as follows:
\begin{enumerate}

\item For $x \leq \alpha t$,
\begin{equation*}
\psi(\eta) = {- x \eta + 4^m (2m+1) t \eta {P_m(\eta^2, 1, u^*, u^*) \over \eta^2 - u^*}
    \over \sqrt{1 - \eta^2}} \ .
\end{equation*}

\item For $\alpha t < x < \beta t$,
\begin{equation*}
\psi(\eta) = \left \{ \begin{matrix} - {-x \eta P_0(\eta^2, 1, u_2, u_3) + 4^m (2m+1) t \eta P_m(\eta^2, 1, u_2, u_3)
\over \sqrt{(1 - \eta^2) (u_2 - \eta^2)(u_3 - \eta^2)}} & \quad 0 < \eta < \sqrt{u_3} \\
0 & \quad \sqrt{u_3} < \eta < \sqrt{u_2} \\
{-x \eta P_0(\eta^2, 1, u_2, u_3) + 4^m (2m+1) t \eta P_m(\eta^2, 1, u_2, u_3)
\over \sqrt{(1 - \eta^2) (\eta^2 - u_2)(\eta^2 - u_3)}} & \quad \sqrt{u_2} < \eta < 1 \ ,\end{matrix} \right. 
\end{equation*}
where $P_0$ and $P_m$ are defined in (\ref{eq15}) and $u_2$ and $u_3$ are determined by equations (\ref{ws1}).

\item For $\beta t < x < 4^m t$,
\begin{equation*}
\psi(\eta) = \left \{ \begin{matrix} 0 & 0 < \eta < \sqrt{u_2} \\
{-x P_0(\eta^2, 1, u_2, 0) + 4^m (2m+1) t P_m(\eta^2, 1, u_2, 0)
\over \sqrt{(1 - \eta^2) (\eta^2 - u_2)}} &  \sqrt{u_2} < \eta < 1 \ , \end{matrix} \right. 
\end{equation*}
where $u_2$ is determined by (\ref{ws2}).

\item For $x \geq 4^m t$, $$\psi(\eta) \equiv 0 \ .$$
\end{enumerate}

\end{thm}

\begin{proof}

We extend the function $\psi$ defined on $[0,1]$ to the entire real line by setting
$\psi(\eta) = 0$ for $\eta > 1$ and taking $\psi$ to be odd. In this way, the operator
$L$ is connected to the Hilbert transform $H$ on the real line \cite{lax}:
$$L \psi(\eta) = \int_0^{\eta} H \psi(\mu) d \mu \quad \mbox{where} \  H \psi (\eta) = {1 \over \pi}
P.V. \int_{- \infty}^{+ \infty} {\psi(\mu) \over \eta - \mu} d \mu \ .$$

We verify case (4) first. Clearly $\psi(\eta) = 0$ satisfies the constraint (\ref{con}).
We now check the variational conditions (\ref{con1}-\ref{con2}). Since $\psi=0$,
$$L \psi = 0 \leq x \eta - 4^m t \eta^{2m+1} \ ,$$
where the inequality follows from $x \geq 4^m t$ and $0 \leq \eta \leq 1$. Hence,
variational conditions (\ref{con1}-\ref{con2}) are satisfied.

Next we consider case (1). We write $\psi(\eta)$ as the real part of $g_1(\eta)$ for real $\eta$,
where
$$g_1= \sqrt{-1}(x - 4^m (2m+1) t \eta^{2m}) + { \sqrt{-1} [-x \eta + 4^m (2m+1) t \eta 
{P_m(\eta^2, 1, u^*, u^*) \over \eta^2 - u^*}
] \over \sqrt{\eta^2 -1}} \ .$$
The function $g_1$ is analytic in the upper
half complex plane $Im (\eta) > 0$ and $g_1(\eta) \approx O(1/\eta^2)$ for large $|\eta|$ 
in view of the expansion (\ref{eq15}) for $P_m$.
Hence, $H \psi (\eta) = Im [g_1(\eta)] = x - 4^m (2m+1) t \eta^{2m}$ on $0 \leq \eta \leq 1$, where $H$ is the
Hilbert transform \cite{lax}. We then have for $0 \leq \eta \leq 1$
$$L \psi (\eta) = \int_0^{\eta} H \eta(\mu) d \mu = x \eta - 4^m t \eta^{2m + 1} \ ,$$
which shows that the variational conditions (\ref{con1}) and (\ref{con2}) are satisfied. 

To prove (\ref{con}), we first claim that
\begin{equation}
\label{1}
- \alpha + 4^m (2m+1) {P_m(\eta^2, 1, u^*, u^*) \over \eta^2 - u^*} \geq 0 \ ,
\end{equation}
for $0 \leq \eta \leq 1$. To see this, we use (\ref{tr}) and (\ref{q}) to 
calculate $\alpha =
\mu_2^{(m)}(1, u^*, u^*)$ and (\ref{P}) and (\ref{Q}) to evaluate $P_m(1, u^*, u^*)$.
The left hand side of (\ref{1}) equals
\begin{align}
2(\eta^2 - 1)&(\eta^2 - u^*) \Phi(\eta^2, 1, u^*, u^*) +
2(\eta^2 - u^*) [ {\p q(1, u^*, u^*) \over \p u_1} + {\p q(1, u^*, u^*) \over \p u_2}
+ {\p q(1, u^*, u^*) \over \p u_3} ]  \no \\
 &= 2(\eta^2 -1) [C_m U(\eta^2, 1) - C_m U(u^*, 1)] + 2(\eta^2 - u^*) C_m
U(u^*, 1) \no \\
 &= 2[(1-u^*) C_mU(u^*, 1) - (1-\eta^2) C_m U(\eta^2,1)] \ , \label{2} 
\end{align}
where in the first equality we have used (\ref{divq}) and the identity (cf. (3.28) of \cite{gra})
$$\Phi(\eta^2, 1, u^*, u^*) = {C_m U(\eta^2, 1) - C_m U(u^*, 1) \over \eta^2 - u^*}
\ .$$
In view of (\ref{U}), we have
$${d \over d (\eta^2)} [(1-\eta^2) C_m U(\eta^2,1)] =
C_m[(1-\eta^2) {\p U(\eta^2, 1) \over \p \eta^2} -  U(\eta^2, 1)] 
= p_1(1, \eta^2, \eta^2) \ .$$
This derivative is positive when $\eta^2 < u^*$ and negative
when $\eta^2 > u^*$ because $p_1$ is so according to Lemma \ref{u*}. 
Therefore, $(1-\eta^2) C_m U(\eta^2,1)$ has a maximum 
at $\eta^2 = u^*$; i.e., $(1-\eta^2) C_m U(\eta^2,1) \leq 
(1-u^*) C_m U(u^*,1)$
for $0 \leq \eta^2 \leq 1$. This together with (\ref{2}) proves (\ref{1}).

It follows from $x \leq \alpha t$ and inequality (\ref{1}) that $\psi \geq 0$.
Hence, the constraint (\ref{con}) is verified.

We now turn to case (2). By Lemma 3.5, the last two equations of (\ref{ws1}) determine
$u_2$ and $u_3$ as functions of the self-similarity variable $x/t$ in the interval
$\alpha \leq x/t \leq \beta$.

We write $\psi = Re \left(g_2\right) $ for real $\eta$, where
$$g_2 = \sqrt{-1}(x - 4^m (2m+1) t \eta^{2m}) + { \sqrt{-1} [-x \eta P_0(\eta^2, 1, u_2, u_3) + 
4^m (2m+1) t \eta P_m(\eta^2, 1, u_2, u_3)
]  \over \sqrt{\eta^2 -1)(\eta^2 - u_2)(\eta^2 - u_3)}} \ .$$

The function $g_2$ is analytic in $Im (\eta) > 0$ and $g_2(\eta) \approx O(1/\eta^2)$ for large $|\eta|$
in view of the asymptotics (\ref{eq16}) for $P_0$ and $P_m$. Hence, taking the imaginary part of $g_2$
yields
$$ H \psi(\eta) = \left \{ \begin{matrix} x - 4^m (2m+1) t \eta^{2m} & 0 < \eta < \sqrt{u_3} \\
       x - 4^m (2m+1) t \eta^{2m} - {[-x P_0(\eta^2, 1, u_2, 0) + 4^m (2m+1) t P_m(\eta^2, 1, u_2, 0)] \eta
\over \sqrt{(1 - \eta^2) (u_2 - \eta^2)(\eta^2 - u_3)}} &  \sqrt{u_3} < \eta < \sqrt{u_2} \\
 x - 4^m (2m+1) t \eta^{2m} & \sqrt{u_2} < \eta < 1 \ . \end{matrix} \right. $$
We then have
\begin{equation}
\label{Lpsi}
L \psi (\eta) =  \left \{ \begin{matrix} x \eta - 4^m t \eta^{2m+1} & 0 < \eta < \sqrt{u_3} \\
 x \eta - 4^m t \eta^{2m+1} - \int_{\sqrt{u_3}}^{\eta} {[-x P_0 + 4^m (2m+1) t P_m] \mu \over \sqrt{(1 - \mu^2)
(u_2 - \mu^2) (\mu^2 - u_3)} } d \mu & \sqrt{u_3} < \eta < \sqrt{u_2} \\
 x \eta - 4^m t \eta^{2m+1} & \sqrt{u_2} < \eta < 1 \ , \end{matrix} \right. 
\end{equation}
where we have used
\begin{equation*}
\int_{\sqrt{u_3}}^{\sqrt{u_2}} {[-x P_0 + 4^m (2m+1) t P_m] \mu \over \sqrt{(1 - \mu^2)
(u_2 - \mu^2) (\mu^2 - u_3)} } d \mu = 0 \ ,
\end{equation*}
which is a consequence of (\ref{eq17}) for $P_0$ and $P_m$.

To verify (\ref{con}), we derive an integral formula for $\psi$. We use (\ref{Pm0}) and (\ref{ws1}) to calculate
\begin{equation*}
\psi_x(\eta) = \left \{ \begin{matrix}  {\eta P_0(\eta^2, 1, u_2(x,t), u_3(x,t)) 
\over \sqrt{(1 - \eta^2) (u_2(x,t) - \eta^2)(u_3(x,t) - \eta^2)}} & \quad 0 < \eta < \sqrt{u_3(x,t)} \\
0 & \quad \sqrt{u_3(x,t)} < \eta < \sqrt{u_2(x,t)} \\
{- \eta P_0(\eta^2, 1, u_2(x,t), 0)
\over \sqrt{(1 - \eta^2) (\eta^2 - u_2(x,t))(\eta^2 - u_3(x,t))}} & \quad \sqrt{u_2(x,t)} < \eta < 1 \ . \end{matrix} \right. 
\end{equation*}
Integrating yields
\begin{equation}
\label{int}
\psi(\eta) = \left \{ \begin{matrix}  - \int_{x}^{\chi(\eta^2) t} {\eta P_0(\eta^2, 1, u_2(y,t), u_3(y,t)) 
\over \sqrt{(1 - \eta^2) (u_2(y,t) - \eta^2)(u_3(y,t) - \eta^2)}} d y & \quad 0 < \eta < \sqrt{u_3(x,t)} \\
0 & \quad \sqrt{u_3(x,t)} < \eta < \sqrt{u_2(x,t)} \\
\int_x^{\chi(\eta^2) t} {\eta P_0(\eta^2, 1, u_2(y,t), 0) 
\over \sqrt{(1 - \eta^2) (\eta^2 - u_2(y,t))(\eta^2 - u_3(y,t))}} d y
& \quad \sqrt{u_2(x,t)} < \eta < 1 \ , \end{matrix} \right. 
\end{equation}
where $\chi(\eta^2)$ is defined in Figure 1.
The polynomial $P_0$ of (\ref{eq15}) is linear in $\eta^2$ and has a zero for $u_3 < \eta^2 < u_2$
because of (\ref{eq17}). $P_0$ must be positive
for $\eta^2 > u_2$ and negative for $\eta^2 < u_3$.  This combined with (\ref{int}) proves $\psi \geq 0$; so
(\ref{con}) is verified.

We now continue to verify the variational conditions (\ref{con1}) and (\ref{con2}).
Again, we use (\ref{Pm0}) and (\ref{ws1}) to calculate
\begin{gather*}
{\p \over \p x} \left( {[- x P_0(\mu^2, 1, u_2(x,t), u_3(x,t)) + 4^m (2m+1) t P_m(\mu^2, 1, u_2(x,t), u_3(x,t))]\mu
\over \sqrt{(1 - \mu^2)(u_2(x,t) - \mu^2)(\mu^2 - u_3^2(x,t))}} \right) \\
 ={ - \mu P_0(\mu^2, 1, u_2(x,t), u_3(x,t))
\over \sqrt{(1 - \mu^2)(u_2(x,t) - \mu^2)(\mu^2 - u_3^2(y,t))}} \ . 
\end{gather*}
Integrating yields 
\begin{gather*}
{[- x P_0(\mu^2, 1, u_2(x,t), u_3(x,t)) + 4^m (2m+1) t P_m(\mu^2, 1, u_2(x,t), u_3(x,t))]\mu
\over \sqrt{(1 - \mu^2)(u_2(x,t) - \mu^2)(\mu^2 - u_3^2(x,t))}} \\
 = - \int_{\chi(\mu^2) t }^x 
{ \mu P_0(\mu^2, 1, u_2(y,t), u_3(y,t)) \over \sqrt{(1 - \mu^2)(u_2(y,t) - \mu^2)(\mu^2 - u_3^2(y,t))}} d y 
\end{gather*}
for $\sqrt{u_3(x,t)} < \mu < \sqrt{u_2(x,t)}$. The single integral in (\ref{Lpsi}) can then be written as a double integral.
After interchanging integrals and using (\ref{eq17}) for $P_0$, the double integral is simplified as
$$- \int_{\chi(\eta^2)t}^x \int_{\sqrt{u_3(y,t)}}^{\eta} {\mu P_0(\mu^2, 1, u_2(y,t), u_3(y,t)) \over
\sqrt{(1 - \mu^2)(u_2(y,t) - \mu^2)(\mu^2 - u_3^2(y,t))} } \ d \mu d y \ .$$
The polynomial $P_0$ is linear in $\mu^2$ and has a zero for $u_3(y,t) < \mu^2 < u_2(y,t)$. In view of (\ref{eq17})
for $P_0$, we must have 
$$\int_{\sqrt{u_3(y,t)}}^{\eta} {\mu P_0(\mu^2, 1, u_2(y,t), u_3(y,t)) \over
\sqrt{(1 - \mu^2)(u_2(y,t) - \mu^2)(\mu^2 - u_3^2(y,t))} } \ d \mu < 0 $$
for $\sqrt{u_3(y,t)} < \eta < \sqrt{u_2(y,t)}$. Hence, the integral in (\ref{Lpsi}) is positive and
this verifies the variational conditions (\ref{con1}) and (\ref{con2}).

We finally consider case (3). By Lemma 3.6, the second equation of (\ref{ws2}) determines $u_2$
as an increasing function of $x/t$ in the interval $\beta \leq x/t \leq 4^m$.

We write $\psi = Re (g_3)$ for real $\eta$, where
$$g_3 = \sqrt{-1}(x - 4^m (2m+1) t \eta^{2m}) + { \sqrt{-1} [-x P_0(\eta^2, 1, u_2, 0) + 4^m (2m+1) t P_m(\eta^2, 1, u_2, 0)
]  \over \sqrt{(\eta^2 -1)(\eta^2 - u_2)}} \ .$$
The function $g_3$ is analytic in $Im (\eta) > 0$ and $g_3(\eta) \approx O(1/\eta^2)$ for large $|\eta|$
in view of the asymptotics (\ref{eq16}) for $P_0$ and $P_m$. Hence, taking the imaginary part of $g_3$
yields
$$H \psi(\eta) = \left \{ \begin{matrix} x - 4^m (2m+1) t \eta^{2m} - {-x P_0(\eta^2, 1, u_2, 0) + 4^m (2m+1)P_m(\eta^2, 1, u_2, 0)
\over \sqrt{(1 - \eta^2)(u_2 - \eta^2)}} & 0 < \eta < \sqrt{u_2} \\
x - 4^m (2m+1) t \eta^{2m}  & \sqrt{u_2} < \eta < 1 \ . \end{matrix} \right. $$
We then have
$$L \psi (\eta) = \left \{ \begin{matrix} x \eta - 4^m t \eta^{2m+1} - \int_0^{\eta} {-x P_0 + 4^m(2m+1) t P_m
\over \sqrt{(1 - \mu^2)(u_2 - \mu^2)}} d \mu & 0< \eta < \sqrt{u_2} \\
 x \eta - 4^m t \eta^{2m+1} & \sqrt{u_2} < \eta < 1 \ , \end{matrix} \right. $$
where we have used
\begin{equation*}
\int_0^{\sqrt{u_2}} {-x P_0(\mu^2, 1, u_2, 0) + 4^m(2m+1) t P_m(\mu^2, 1, u_2, 0) \over
\sqrt{(1 - \mu^2)(u_2 - \mu^2)}} d \mu = 0 \ ,
\end{equation*}
which is a consequence of (\ref{eq17}) for $P_0$ and $P_m$.

The variational conditions (\ref{con1}-\ref{con2})
and the constraint (\ref{con}) can be verified using the method for case (2).

\end{proof}

{\bf Acknowledgments.} 
We thank Tamara Grava for some of the ideas leading to the proof of Lemma \ref{u*}.
V.P. was supported in part by NSF Grant DMS-0135308. F.-R. T. was supported in part by
NSF Grant DMS-0404931.

\bibliographystyle{amsplain}

\end{document}